\RequirePackage{amsmath}
\documentclass[runningheads,a4paper]{llncs}

\usepackage{amssymb,amsmath,tikz,changepage}
\usepackage{graphicx}

\usepackage{url}
\usepackage{enumerate}

\makeatletter
\newcommand{\text@hyphens}{\mathcode`\-=`\-\relax}
\newcommand{\id}[1]{\ensuremath{\mathit{\text@hyphens#1}}}
\makeatother

\begin{document}

\mainmatter  
\title{On Black-Box Transformations in Downward-Closed Environments}
\titlerunning{On Black-Box Transformations in Downward-Closed Environments}
\author{Warut Suksompong%
}
\authorrunning{W. Suksompong}
\institute{Department of Computer Science, Stanford University\\
353 Serra Mall, Stanford, CA 94305, USA\\
\email{warut@cs.stanford.edu}\\
}
\maketitle

\begin{abstract}
Black-box transformations have been extensively studied in algorithmic mechanism design as a generic tool for converting algorithms into truthful mechanisms without degrading the approximation guarantees. While such transformations have been designed for a variety of settings, Chawla et al. showed that no fully general black-box transformation exists for single-parameter environments. In this paper, we investigate the potentials and limits of black-box transformations in the prior-free (i.e., non-Bayesian) setting in \emph{downward-closed} single-parameter environments, a large and important class of environments in mechanism design. On the positive side, we show that such a transformation can preserve a constant fraction of the welfare at every input if the private valuations of the agents take on a constant number of values that are far apart, while on the negative side, we show that this task is not possible for general private valuations.
\end{abstract}

\section{Introduction}

Mechanism design is a science of rule-making. Its goal is to design rules so that individual strategic behavior of the agents leads to desirable global outcomes. Algorithmic mechanism design, one of the initial and most well-studied branches of algorithmic game theory, studies the tradeoff between optimizing the global outcome, respecting the incentive constraints for individual agents, and maintaining the computational tractability of the mechanism \cite{NisanRo01}. A major line of work in algorithmic mechanism design involves taking a setting where the optimization problem is computationally intractable, and designing computationally tractable mechanisms that yield a good global outcome and such that the agents have a truth-telling incentive. Ideally, the mechanisms would match the best-known approximation guarantees for computationally tractable optimization algorithms in that setting. In other words, we want to obtain truthfulness from agents in as many settings as possible without having to pay for more computation.

In the past two decades, this goal of algorithmic mechanism design has been met in a wide range of prior-free as well as Bayesian settings. For instance, Briest et al. \cite{BriestKrVo11} showed how to transform pseudopolynomial algorithms for several problems, including knapsack, constrained shortest path, and scheduling, into monotone fully polynomial time approximation schemes (FPTAS), which lead to efficient and truthful auctions for these problems. Lavi and Swamy \cite{LaviSw11} constructed a general reduction technique via linear programming that applies to a wide range of problems. The widespread success of designing computationally tractable mechanisms with optimal approximation guarantees has raised the question of whether there exists a generic method for transforming any computationally tractable algorithm into a computationally tractable mechanism without degrading the approximation guarantee. Such a method would not be allowed access to the description of the algorithm but instead would only be able to query the algorithm at specific inputs, and is therefore known as a ``black-box transformation''.

An important work that demonstrates a limit of the powers of black-box transformations was done by Chawla et al. \cite{ChawlaImLu12}, who showed among other things that no fully general black-box transformation exists for single-parameter environments in the prior-free setting. In particular, for any transformation, there exists an algorithm (along with a feasibility set) such that the transformation degrades the approximation ratio of the algorithm by at least a polynomial factor. The result holds even when the private valuations can take on only two values; Chawla et al. provided a construction with two private valuations $l<h$ satisfying $h/l=n^{7/10}$, where $n$ is the number of agents. Pass and Seth \cite{PassSe14} extended this result by allowing the transformation access to the feasibility set while assuming the existence of cryptographic one-way functions. 

Even though no fully general black-box transformation exists for single-parameter environments, it is still conceivable that there are transformations that work for certain large subclasses of such environments. One important subclass, which is the main subject of our paper, is that of \emph{downward-closed} environments, i.e., environments in which any subset of a feasible allocation is also feasible. The construction used by Chawla et al. \cite{ChawlaImLu12}, later built upon by Pass and Seth \cite{PassSe14}, relies heavily on the non-downward-closedness of the feasibility set. The construction only includes three feasible allocations, and it is crucial that the transformation cannot arbitrarily ``round down'' the allocations as it would be able to if the feasibility set were downward-closed. Since downward-closed environments occur in a wide variety of settings in mechanism design, including knapsack auctions and combinatorial auctions, we find the question that we study to be a natural and important one. We consider such settings and assume, crucially, that the black-box transformation is aware that the feasible set is downward-closed. As a result, when the transformation makes a query to the algorithm, it can potentially learn many more feasible allocations than merely the one it obtains. In this paper, we investigate the potentials and limits of black-box transformations when they are endowed with this extra power.

\subsection{Our results}

Throughout the paper, we consider the prior-free (i.e., non-Bayesian) setting. In Section \ref{sec:negativeresults}, we show the limits of black-box transformations in downward-closed environments. We prove that such transformations cannot preserve the full welfare at every input, even when the private valuations can take on only two arbitrary values (Theorem \ref{thm:100pointwise}). Preserving a constant fraction of the welfare pointwise is impossible if the ratio between the two values $l<h$ is sublinear, i.e., $h/l\in O(n^\alpha)$ for $\alpha\in[0,1)$, where $n$ is the number of agents (Theorems \ref{thm:constratioconstpointwise} and \ref{thm:nonconstratioconstpointwise}), while preserving the approximation ratio is also impossible if the values are within a constant factor of each other and the transformation is restricted to querying inputs of Hamming distance $o(n)$ away from its input (Theorem \ref{thm:worstcasepointwisehamming}). 

In Section \ref{sec:positiveresults}, we show the powers of black-box transformations in downward-closed environments. We prove that when the private valuations can take on only a constant number of values, each pair of values separated by a ratio of $\Omega(n)$, it becomes possible for a transformation to preserve a constant fraction of the welfare pointwise, and therefore the approximation ratio as well (Theorem \ref{thm:algotwo}). The same is also true if the private valuations are all within a constant factor of each other (Theorem \ref{thm:constalloc}). Combined with the negative results, this gives us a complete picture of constant-fraction welfare-preserving transformations for multiple input values. Not only are these results interesting in their own right, but they also demonstrate the borders of the negative results that we can hope to prove.

The results are summarized in Table \ref{table:summary} for the case where the private valuations can take on two values, but they can be generalized to any constant-size range of private valuations as well.

    \begin{table*}
\begin{center}
    \begin{tabular}{ | p{2.4cm} | p{2cm} | p{2cm} | p{3.3cm} | p{2cm} | }
    \hline
     & $100\%$ pointwise & Constant fraction pointwise & $100\%$ approx ratio & Constant fraction approx ratio \\ \hline

    $h/l=n^{7/10}$; $\mathcal{F}$ unknown, not downward-closed & No \cite{ChawlaImLu12} & No \cite{ChawlaImLu12} & No \cite{ChawlaImLu12} & No \cite{ChawlaImLu12} \\ \hline

    $h/l\in\Omega(n)$; $\mathcal{F}$ known, downward-closed & No  (Theorem \ref{thm:100pointwise}) & Yes  (Theorem \ref{thm:algotwo}) & ? & Yes  (Theorem \ref{thm:algotwo}) \\ \hline

    $h/l\in\Theta(1)$; $\mathcal{F}$ known, downward-closed & No  (Theorem \ref{thm:100pointwise}) & No  (Theorem \ref{thm:constratioconstpointwise}) & No if restricted to Hamming distance $f(n)\in o(n)$, $\mathcal{F}$ unknown  (Theorem \ref{thm:worstcasepointwisehamming}) & Yes  (Theorem \ref{thm:constalloc}) \\ \hline

    $h/l\in O(n^\alpha)$ for $\alpha\in(0,1)$; $\mathcal{F}$ known, downward-closed & No  (Theorem \ref{thm:100pointwise}) & No (Theorem \ref{thm:nonconstratioconstpointwise}) & ? & ? \\ \hline
    \end{tabular}
    \caption{Summary of our results for the case where the private valuations take on two values $l<h$. The results can be generalized to any constant-size range of private valuations.}
    \label{table:summary}
\end{center}
    \end{table*}

\subsection{Related work}

Besides the works already mentioned, black-box transformations have been obtained in a variety of other prior-free and Bayesian settings. In the prior-free setting, Goel et al. \cite{GoelKaWa10} presented a reduction for symmetric single-parameter problems with a logarithmic loss in approximation, and later Huang et al. \cite{HuangWaZh11} improved the reduction to obtain arbitrarily small loss. Dughmi and Roughgarden \cite{DughmiRo14} designed a reduction for the class of multi-parameter problems that admit an FPTAS and can be encoded as a packing problem, while Babaioff et al. \cite{BabaioffLaPa09} considered reductions for single-valued combinatorial auction problems.  
Reductions that preserve the approximation guarantees have also been obtained in the single-parameter Bayesian setting by Hartline and Lucier \cite{HartlineLu10}, and their work was later extended to multi-parameter settings by Bei and Huang \cite{BeiHu11}, Cai et al. \cite{CaiDaWe13}, and Hartline et al. \cite{HartlineKlMa15}. 

\section{Preliminaries}

We will be concerned with single-parameter environments. Such an environment consists of some number $n$ of agents. Each agent $i$ has a private valuation $v_i\in\mathbb{R}$, its value ``per unit of stuff'' that it gets. In addition, there is a feasibility set $\mathcal{F}$, which specifies the allocations that can be made to the agents. Each element of $\mathcal{F}$ is a vector $(x_i)_{i=1}^n$, where $x_i\in\mathbb{R}$ denotes the ``amount of stuff'' given to agent $i$. For instance, in single-item auctions, $\mathcal{F}$ consists of the vectors with $x_i\in\{0,1\}$ and $\sum_{i=1}^nx_i=1$. A more general and well-studied type of auctions is called knapsack auctions, in which each agent is endowed with a public size $w_i$ along with its private valuation $v_i$, and the seller has some public capacity $W$. The feasibility set of a knapsack auction consists of the vectors with $x_i\in\{0,1\}$ and $\sum_{i=1}^nw_ix_i\leq W$. In this paper, we will assume that the feasibility set $\mathcal{F}$ is \emph{downward-closed}, which means that if we take an allocation and decrease the amount of stuff given to one of the agents, then the resulting allocation is also feasible. Downward-closedness is an assumption that holds in many natural settings, including the aforementioned auctions.

\paragraph{Algorithms} An \emph{algorithm} (or \emph{allocation rule}) $\mathcal{A}$ is a function that takes as input a valuation vector $\textbf{v}=(v_i)_{i=1}^n$ and outputs an allocation $\textbf{x}=(x_i)_{i=1}^n$. We will consider the social welfare objective---the \emph{welfare} of $\mathcal{A}$ at $\textbf{v}$ is given by $\textbf{v}\cdot\textbf{x}=v_1x_1+\dots+v_nx_n$, where $\textbf{x}\in\mathcal{F}$ is the allocation that $\mathcal{A}$ returns at $\textbf{v}$. We denote by $OPT_\mathcal{F}(\textbf{v})$ the maximum welfare at valuation vector $\textbf{v}$ over all allocations in $\mathcal{F}$. The \emph{(worst-case) approximation ratio} of $\mathcal{A}$ is given by $approx_\mathcal{F}(\mathcal{A})=\min_{\textbf{v}}\frac{\mathcal{A}(\textbf{v})}{OPT_\mathcal{F}(\textbf{v})}$, where we slightly abuse notation and use $\mathcal{A}(\textbf{v})$ to denote the the allocation returned by $\mathcal{A}$ at $\textbf{v}$ as well as the welfare of that allocation at $\textbf{v}$. Note that by definition, $approx_\mathcal{F}(\mathcal{A})\leq 1$ for all $\mathcal{F}$ and $\mathcal{A}$.

\paragraph{Transformations} A \emph{transformation} $\mathcal{T}$ is an algorithm that has black-box access to some other algorithm $\mathcal{A}$, i.e., it can make queries to $\mathcal{A}$. In each query, $\mathcal{T}$ specifies a valuation vector $\textbf{v}$ and obtains the allocation that $\mathcal{A}$ returns at $\textbf{v}$. We write $\mathcal{T}(\mathcal{A})$ for a transformation $\mathcal{T}$ with access to the algorithm $\mathcal{A}$. Importantly, we assume that $\mathcal{T}$ has the knowledge that the feasibility set $\mathcal{F}$ is downward-closed. For the strongest possible negative results, we assume whenever possible that (i) $\mathcal{T}$ has knowledge of $\mathcal{F}$, i.e., it can make a polynomial number of queries to ask whether a particular allocation belongs to $\mathcal{F}$, and (ii) $\mathcal{T}$ is adaptive, i.e., it can adjust its next query based on the responses it received for previous queries. For strongest positive results, our transformation $\mathcal{T}$ does not make queries about $\mathcal{F}$ and is also not adaptive. We will be clear about our assumptions on $\mathcal{T}$ for each result.

\paragraph{Mechanisms} A \emph{mechanism} is a procedure that consists of eliciting declared private valuations $(b_i)_{i=1}^n$ from the agents, and then applying an allocation rule and a payment rule on the elicited valuations. The allocation rule determines the allocation $(x_i)_{i=1}^n$ and the payment rule determines the prices $(p_i)_{i=1}^n$ to charge the agents. We are interested in transformations that, when coupled with any algorithm, lead to \emph{truthful} mechanisms, meaning that it is always in the best interest for each agent $i$ to declare the true valuation $v_i$ to the mechanism, no matter what the other agents do. A seminal result by Myerson \cite{Myerson81} states that an allocation rule can be supplemented with a payment rule to yield a truthful mechanism exactly when the allocation rule is \emph{monotone}. Monotonicity of an allocation rule means that if an agent increases its declared valuation while the declared valuations of the remaining agents stay fixed, then the agent is allocated at least as much stuff as before by the allocation rule. Therefore, the transformations that yield truthful mechanisms are exactly the ones that constitute a monotone allocation rule for any algorithm. 

\paragraph{Properties of transformations} We call a transformation $\mathcal{T}$ \emph{monotone} if $\mathcal{T}(\mathcal{A})$ is a monotone allocation rule for any algorithm $\mathcal{A}$. Furthermore, $\mathcal{T}$ is called \emph{welfare-preserving} if $\mathcal{T}(\mathcal{A})$ preserves the welfare of $\mathcal{A}$ at every input for any algorithm $\mathcal{A}$, and \emph{constant-fraction welfare-preserving} if $\mathcal{T}(\mathcal{A})$ preserves a constant fraction of the welfare of $\mathcal{A}$ at every input for any algorithm $\mathcal{A}$. Similarly, $\mathcal{T}$ is \emph{approximation-ratio-preserving} if $\mathcal{T}(\mathcal{A})$ preserves the approximation ratio of $\mathcal{A}$ for any algorithm $\mathcal{A}$, and \emph{constant-fraction approximation-ratio-preserving} if $\mathcal{T}(\mathcal{A})$ preserves a constant fraction of the approximation ratio of $\mathcal{A}$ for any algorithm $\mathcal{A}$. Note that a (constant-fraction) welfare-preserving transformation is also (constant-fraction) approximation-ratio-preserving.

\section{Negative Results}
\label{sec:negativeresults}

In this section, we consider the limits of black-box transformation in downward-closed environments. First, we show that no monotone black-box transformation preserves, up to a constant factor, the welfare of any original algorithm $A$ pointwise. We then show that if a monotone black-box transformation preserves the approximation ratio of any given input algorithm $A$, then on some input $\textbf{v}$ it must query $A$ on an input that has Hamming distance $\Omega(n)$ from $\textbf{v}$.

\subsection{Welfare-preserving transformations} 

We begin by considering the highest possible benchmark for the transformation: preserving the full welfare of any algorithm at every input. Our first theorem shows that this benchmark is impossible to fulfill even when the private valuations can take on only two arbitrary values.

\begin{theorem} 
\label{thm:100pointwise}
Let $l<h$ be arbitrary values (possibly depending on $n$). There does not exist a polynomial-time, monotone, welfare-preserving transformation, even when the transformation is allowed to be randomized and adaptive and make a polynomial number of queries to $\mathcal{F}$.
\end{theorem}

Before we go into the formal proof, we give a high-level intuition. We handle the easier case of deterministic and non-adaptive transformations before moving to general transformations. We will consider a class of algorithms from which one algorithm $\mathcal{A}$ is selected randomly. For each algorithm $\mathcal{A}$, our feasibility set will contain two maximal allocations $C$ and $D$; the only allocations in $\mathcal{F}$ are those that are subsets of $C$ or $D$. The allocation $C$ is only returned at a ``special input'' $B_1$, and the allocation itself as well as the special input depends on the algorithm $\mathcal{A}$ we choose from the class. At any other input, the allocation $D$ is returned. Using monotonicity of the transformation, we will show that at an input that is ``far away'' from $B_1$, the transformation still needs to know the allocation $C$ in order to preserve the full welfare. However, because of the randomization, the probability the transformation can discover either the allocation $C$ or the special input $B_1$ when it is given the faraway input is exponentially low, meaning that the transformation cannot achieve its goal.

\begin{proof}

Assume first that the transformation $\mathcal{T}$ is deterministic and non-adaptive. Suppose that the input is of length $n=4m$. The algorithm $\mathcal{A}$ will be chosen randomly. To begin, we define the preliminary algorithm $\mathcal{A}$ as follows.
\begin{itemize}
\item At input $B_1=\overbrace{hh\dots h}^{2m}\overbrace{hh\dots h}^m\overbrace{ll\dots l}^m$, $\mathcal{A}$ returns output\\$C=\overbrace{10110\dots 1}^{m+1 \text{ are 1's}}\overbrace{00\dots 0}^m\overbrace{11\dots 1}^m$, where the $m+1$ 1's in the first $2m$ positions are uniformly randomized. Note that the randomization is in the step of choosing the algorithm $\mathcal{A}$, but the resulting algorithm $\mathcal{A}$ itself is a deterministic algorithm. We call this input the \textit{special input};

\item At any other input, $\mathcal{A}$ returns $D=\overbrace{00\dots 0}^{2m}\overbrace{11\dots 1}^m\overbrace{11\dots 1}^m$.
\end{itemize}
In the real algorithm $\mathcal{A}$, we permute uniformly at random the last $3m/2$ positions of the inputs as well as the corresponding allocations. Again, this permutation is only for choosing the (deterministic) algorithm $\mathcal{A}$.

Consider any algorithm $\mathcal{A}$ that we might choose, and assume without loss of generality that in the special input of this algorithm, the $l$'s are in the last $m$ positions. To preserve the welfare at input $B_1$, $\mathcal{T}$ must return $C=\mathcal{A}(B_1)$ itself, since returning any strict subset of $C$ or returning $D$ (or any subset of $D$) would yield a lower welfare.

Next, consider the input $B_2=\overbrace{hh\dots h}^{2m}\overbrace{hh\dots hl}^m\overbrace{ll\dots l}^m$, with the only change from $B_1$ being in the rightmost position of the middle block. By monotonicity, $\mathcal{T}$ must return 0 in that position. In order to preserve the welfare at $B_2$, $\mathcal{T}$ must return a subset of $C$, since otherwise it would have to return a strict subset of $D$, which would yield a lower welfare.

Now, consider inputs $B_3=\overbrace{hh\dots h}^{2m}\overbrace{hh\dots hll}^m\overbrace{ll\dots l}^m$, $B_4=\overbrace{hh\dots h}^{2m}\overbrace{hh\dots hlll}^m\overbrace{ll\dots l}^m$, and so on with one extra $l$ in each input, up to $B_{m/2+1}=\overbrace{hh\dots h}^{2m}\overbrace{hh\dots h}^{m/2}\overbrace{ll\dots l}^{m/2}\overbrace{ll\dots l}^m$. By a similar argument, $\mathcal{T}$ must return a subset of $C$ at all of these inputs. In particular, $\mathcal{T}$ must return a subset of $C$ at $B_{m/2+1}$. 

In order to preserve the welfare at $B_{m/2+1}$, $\mathcal{T}$ must return at least $m/2$ 1's in the first $2m$ positions. If $\mathcal{T}$ tries to find such an allocation by querying $\mathcal{F}$, then since the positions of the $m+1$ 1's in the first $2m$ positions are chosen randomly, the probability of success for each query is at most $\frac{\binom{m+1}{m/2}}{\binom{2m}{m/2}}<\frac{1}{poly(n)}$. Hence $\mathcal{T}$ will succeed within a polynomial number of queries with low probability. 

Alternatively, $\mathcal{T}$ might try to find the special input $B_1$ by querying $\mathcal{A}$. However, recall that we randomly permute the last $3m/2$ positions of the inputs and their corresponding allocations. In order to find the special input $B_1$, $\mathcal{T}$ must correctly choose $m/2$ out of the $3m/2$ positions to change to $h$. Once again, the probability of success for each query is less than $\frac{1}{poly(n)}$. Hence $\mathcal{T}$ will again succeed within a polynomial number of queries with low probability. Combined with the previous paragraph, this means that $\mathcal{T}$ is unlikely to succeed if it is deterministic and non-adaptive.

Now assume that $\mathcal{T}$ is possibly adaptive. We will make sure that for each ``unsuccessful'' query to $\mathcal{A}$ or $\mathcal{F}$, $\mathcal{T}$ learns no new information. This is already the case for queries to $\mathcal{A}$, as for any unsuccessful query, $\mathcal{T}$ simply finds out the canonical allocation $D$. To prevent $\mathcal{T}$ from learning new information from unsuccessful queries to $\mathcal{F}$, we insert ``fake'' allocations into $\mathcal{F}$. In particular, we insert all allocations with $m/2-1$ 1's in the first $2m$ positions and all $1$'s in the last $m$ positions (before the permutation of indices) into $\mathcal{F}$, as well as subsets of these allocations. As such, a successful query to $\mathcal{F}$ that contains at most $m/2-1$ 1's in the first $2m$ positions does not give $\mathcal{T}$ any useful information.

Finally, assume that $\mathcal{T}$ is allowed to be randomized. Since our algorithm $\mathcal{A}$ is also chosen randomly by uniformly permuting the positions of 1's in the allocation as well as permuting the indices, the probability of success of $\mathcal{T}$ in guessing the special input or an allocation that returns at least $m/2$ 1's in the first $2m$ positions cannot increase even if it randomizes its choices. \qed
\end{proof}

\subsection{Constant-fraction welfare-preserving transformations}

Even though Theorem \ref{thm:100pointwise} shows that it is impossible for a transformation to preserve the full welfare pointwise, it would still be interesting if the transformation can preserve a \emph{constant fraction} of the welfare pointwise. However, as we show in this subsection, it turns out that this weaker requirement is also impossible to satisfy. Our next two theorems show that preserving a constant fraction pointwise is impossible when the ratio $h/l$ is sublinear, i.e., $h/l\in O(n^\alpha)$ for some $\alpha\in[0,1)$. We first consider the case where $h/l$ is constant (Theorem \ref{thm:constratioconstpointwise}), and later generalize to $h/l\in O(n^\alpha)$ for some $\alpha\in[0,1)$ (Theorem \ref{thm:nonconstratioconstpointwise}). Together with Theorem \ref{thm:algotwo}, which exhibits an example of a constant-fraction welfare-preserving transformation when $h/l\in\Omega(n)$, we have a complete picture of constant-fraction welfare-preserving transformations when there are two input values.

\begin{theorem}
\label{thm:constratioconstpointwise}
Let $l<h$ be such that $h/l$ is constant. There does not exist a polynomial-time, monotone, constant-fraction welfare-preserving transformation, even when the transformation is allowed to be randomized and adaptive and make a polynomial number of queries to $\mathcal{F}$.
\end{theorem}

We give an outline of the proof, which bears a resemblance to the proof of Theorem \ref{thm:100pointwise} but contains differences in the execution. We start with a deterministic and non-adaptive transformation $\mathcal{T}$. Our feasibility set will contain two maximal allocations $C$ and $D$, as well as subsets of any of these two allocations. The ``special allocation'' $C$ is only returned at the ``special input'' $B_1$. Both the special allocation and the special input are chosen based on the queries that $\mathcal{T}$ makes to $\mathcal{A}$ and $\mathcal{F}$ at various inputs. At any other input, the allocation $D$ is returned. Using the monotonicity of the transformation, we find another input $B_{m+1}$ far away from $B_1$ where we have to return a subset of $C$ that is not a subset of $D$. By our choice of the special allocation and special input, we ensure that at input $B_{m+1}$, the transformation neither makes a query at $B_1$ nor makes a successful query to $\mathcal{F}$. This implies that $\mathcal{T}$ cannot succeed within a polynomial number of queries. 

\begin{proof}
Assume first that the transformation $\mathcal{T}$ is deterministic and non-adaptive. Suppose that the input is of length $n=m^6+m^3$. Note that the sets $poly(m)$ and $poly(n)$ are identical. 

Let $X$ denote the set of inputs with $m^5$ $l$'s in the first $m^6$ positions followed by $m^3$ $h$'s, and let $Y$ denote the set of inputs with $m^5-m$ $l$'s in the first $m^6$ positions, followed by $m^3$ $h$'s. We have $|X|=\binom{m^6}{m^5}$ and $|Y|=\binom{m^6}{m^5-m}$. Since $|X|>|Y|\cdot poly(n)$, there exists an input in $X$ that is not in the (polynomially long) query list of $\mathcal{T}$ for any input in $Y$. Assume without loss of generality that $B_1$, defined below, is one such input.

Consider the algorithm $\mathcal{A}$ as follows:
\begin{itemize}
\item At input $B_1=\overbrace{hh\dots h}^{m^6-m^5}\overbrace{ll\dots l}^{m^5}\overbrace{hh\dots h}^{m^3}$, $\mathcal{A}$ returns \\$C=\overbrace{00\dots 0}^{m^6-m^5}\overbrace{10110\dots 1}^{m^4 \text{ are 1's}}\overbrace{00\dots 0}^{m^3}$, where the $m^4$ 1's in the $m^5$ positions of the middle block are to be chosen later. We call this input the \textit{special input}, and the corresponding allocation the \textit{special allocation};

\item At any other input, $\mathcal{A}$ returns $D=\overbrace{00\dots 0}^{m^6-m^5}\overbrace{00\dots 0}^{m^5}\overbrace{11\dots 1}^{m^3}$.
\end{itemize}

For large enough $m$, to preserve a constant fraction of the welfare at input $B_1$, $\mathcal{T}$ cannot return a subset of $D$. Hence $\mathcal{T}$ must return a subset of $C=\mathcal{A}(B_1)$.

Consider the input $B_2=\overbrace{hh\dots h}^{m^6-m^5}\overbrace{hll\dots l}^{m^5}\overbrace{hh\dots h}^{m^3}$, with the only change from $B_1$ being in the leftmost position of the middle block. (Here we choose the leftmost position because this position of $C$ contains a 1 in the particular choice of $C$ above; otherwise we choose any position of $C$ that contains a 1.) By monotonicity, $\mathcal{T}$ must return a 1 in the middle block for $B_2$, so it cannot return a subset of $D$. Moreover, for large enough $m$, to preserve a constant fraction of the welfare at $B_2$, $\mathcal{T}$ must return at least $m^2$ 1's in the middle block. In particular, there is still a 1 corresponding to an $l$ in the middle block. 

Similarly, we can define inputs $B_3,B_4,\dots,B_{m+1}$ so that $B_i$ has $i-1$ $h$'s in the middle block and there is still a 1 corresponding to an $l$ in the middle block. For each of these inputs, $\mathcal{T}$ must return at least $m^2$ 1's in the middle block. Note also that $B_{m+1}\in Y$.

Now, the special allocation is at $B_1$, and by our assumption above, $\mathcal{T}$ does not find out by querying $\mathcal{A}$ at $B_1$ when it is presented with $B_{m+1}\in Y$. The only other possibility for $\mathcal{T}$ to discover the special allocation is to query $\mathcal{F}$. There are $\binom{m^5}{m}$ inputs of $Y$ whose first $m^6-m^5$ positions are all $h$'s, and these are the only inputs at which $\mathcal{T}$ can benefit from a ``successful'' query to $\mathcal{F}$. When $\mathcal{T}$ makes a query at each of these inputs, it must pick an allocation with at least $m^2$ $1$'s in the $m^5$ positions. From the perspective of us preventing the transformation $\mathcal{T}$ from achieving its goal, this rules out at most $\binom{m^5-m^2}{m^4-m^2}$ allocations. The total number of possible allocations that we can choose is $\binom{m^5}{m^4}$. Since $\binom{m^5}{m}\cdot\binom{m^5-m^2}{m^4-m^2}\cdot poly(n)<\binom{m^5}{m^4}$, this means that for some choice of 1's in $m^4$ out of the $m^5$ positions in the middle block, $\mathcal{T}$ does not succeed in finding an allocation with at least $m^2$ $1$'s. By making this choice, we ensure that $\mathcal{T}$ cannot succeed within a polynomial number of queries.

Finally, we generalize to adaptive and randomized transformations $\mathcal{T}$ in a similar way as in Theorem \ref{thm:100pointwise}. \qed
\end{proof}

Using a similar construction, we can generalize the impossibility result to the case where $h/l\in O(n^\alpha)$ for any $\alpha\in [0,1)$.

\begin{theorem}
\label{thm:nonconstratioconstpointwise}
Let $l<h$ be such that $h/l\in O(n^\alpha)$ for some $\alpha\in [0,1)$. There does not exist a polynomial-time, monotone, constant-fraction welfare-preserving transformation, even when the transformation is allowed to be randomized and adaptive and make a polynomial number of queries to $\mathcal{F}$.
\end{theorem}

\begin{proof}
We extend the example for the case where $h/l$ is constant (Theorem \ref{thm:constratioconstpointwise}). Let $b>c>d>e>1$ be constants that we will choose later. Suppose that the three blocks have length $m^b, m^c$, and $m^e$, respectively, and that there are $m^d$ 1's in the middle block. We construct inputs $B_1,B_2,\dots,B_{m+1}$ as before. In order for the same argument to go through, we need three conditions:
\begin{enumerate}
\item For the transformation to necessarily return a subset of $C$ at the special input $B_1$, we need $m^d\geq m^e(m^b+m^c+m^e)^\alpha$. This translates roughly to $d\geq e+b\alpha$. Since $\alpha<1$, it is possible to choose such $b>d>e$.

\item For queries to $\mathcal{A}$ to succeed with low probability, we need $\binom{m^b+m^c}{m^c-m}\cdot poly(n)<\binom{m^b+m^c}{m^c}$. This always holds for $b>c>1$.

\item For queries to $\mathcal{F}$ to succeed with low probability, we need $\binom{m^c}{m}\cdot\binom{m^c-m^e}{m^d-m^e}\cdot poly(n)<\binom{m^c}{m^d}$. This translates roughly to $m^{(c-d)m^e-cm}>poly(n)$, which always holds for $c>d>e>1$.
\end{enumerate}
Hence we can choose $b,c,d,e$ so that all three conditions hold, and our proof is complete. \qed
\end{proof}

Note that the examples so far cannot be used to show the non-existence of a monotone (constant-fraction) approximation-ratio-preserving transformation. Indeed, consider the transformation that simply returns the canonical allocation $D$. The points at which this transformation fails to preserve the welfare of the algorithm are points at which the algorithm is optimal, and elsewhere the algorithm is far from optimal, implying that the approximation ratio is preserved.

\subsection{Approximation-ratio-preserving transformations}

In this subsection, we consider a weaker benchmark than preserving full welfare pointwise: preserving the approximation ratio. We show that this benchmark is still impossible to satisfy if we restrict the transformation $\mathcal{T}$ to querying inputs at Hamming distance less than some function $f(n)\in o(n)$ from its input, and disallow $\mathcal{T}$ from querying $\mathcal{F}$.

\begin{theorem}
\label{thm:worstcasepointwisehamming}
Let $l<h$ be such that $h/l$ is constant, and let $f(n)\in o(n)$. There does not exist a polynomial-time, monotone, approximation-ratio-preserving transformation $\mathcal{T}$. The transformation $\mathcal{T}$ is allowed to be randomized and adaptive, but it cannot make queries to $\mathcal{F}$ and can only make queries to $\mathcal{A}$ on inputs that are of Hamming distance less than $f(n)$ from the original input.
\end{theorem}

\begin{proof}
Suppose that the input is of length $n=2m$, and consider the algorithm $\mathcal{A}$ as follows:
\begin{itemize}
\item At any input with at most $m+f(n)$ $h$'s, $\mathcal{A}$ returns $\overbrace{00\dots 0}^m\overbrace{11\dots 1}^m$;

\item At any other input, $\mathcal{A}$ returns $\overbrace{11\dots 1}^m\overbrace{00\dots 0}^m$.
\end{itemize}

One can check that $approx_\mathcal{F}(\mathcal{A})=l/h$. Let $B_1$ denote the input $\overbrace{hh\dots h}^m\overbrace{ll\dots l}^m$. We have $\mathcal{A}(B_1)=\overbrace{00\dots 0}^m\overbrace{11\dots 1}^m$. At input $B_1$, the transformation $\mathcal{T}$ cannot discover the other (undominated) allocation because of the Hamming distance restriction. Hence it must return a subset of $\mathcal{A}(B_1)$. Moreover, since the approximation ratio of $\mathcal{A}$ is worst at input $B_1$, $\mathcal{T}$ must return exactly $\mathcal{A}(B_1)$.

Consider the input $B_2=\overbrace{hh\dots h}^m\overbrace{hll\dots l}^m$, with the only change from $B_1$ being in the leftmost position of the second half. By monotonicity, $\mathcal{T}$ must return a subset of $\mathcal{A}(B_1)$ at $B_2$. Moreover, for large enough $n$, to preserve the approximation ratio, $\mathcal{T}$ must return at least one 1 on $l$ in the second half.

Similarly, we can define inputs $B_3,B_4,\dots,B_{m+2f(n)+1}$ so that $B_i$ has $i-1$ $h$'s in the second half and there is still a 1 corresponding to an $l$ in the second half. A sufficient condition to guarantee a 1 on an $l$ in the second half is that putting 1's on all $h$'s in the second half is not enough to match the approximation ratio $l/h$. That is, $\frac{2f(n)\cdot h}{mh}<\frac{l}{h}$. Since $f(n)\in o(n)$, we can choose $n$ large enough so that this condition is satisfied. 

For each of the inputs $B_3,B_4,\dots,B_{m+2f(n)+1}$, $\mathcal{T}$ must return a subset of $\mathcal{A}(B_1)$. At input $B_{m+2f(n)+1}$, however, $\mathcal{T}$ cannot discover the allocation $\mathcal{A}(B_1)$ because of the Hamming distance restriction. Hence $\mathcal{T}$ cannot succeed. \qed
\end{proof}

If we are only interested in preserving a constant factor of the approximation ratio, then Theorem \ref{thm:constalloc} shows that this is possible in the same setting of $h/l$ constant, and Theorem \ref{thm:algotwo} shows that it is also possible when $h/l\in\Omega(n)$. It is not clear whether a negative result can be obtained when $h/l\in O(n^\alpha)$ for some $\alpha\in(0,1)$.

\section{Positive Results}
\label{sec:positiveresults}

In this section, we consider the powers of black-box transformations in downward-closed environments. We show that when values are either high or low, and the ratio between high and low is $\Omega(n)$, then there is a monotone transformation that gives a constant approximation to the welfare of any given algorithm pointwise, and therefore also preserves the approximation ratio up to a constant factor. This can be generalized to any constant number of values, and the transformation can be modified so that it also preserves full welfare at a constant fraction of the inputs. While these results are of independent interest, they also serve to demonstrate the limitations of extending the negative results in Section \ref{sec:negativeresults}. For the strongest possible results, we exhibit transformations that do not query $\mathcal{F}$ or operate adaptively.

\subsection{Two values}

We begin by showing that when the private valuations take on two values that are far apart, there exists a transformation that preserves a constant fraction of the welfare at each input. This contrasts with the negative result when the values are close to each other (Theorem \ref{thm:nonconstratioconstpointwise}).

\begin{theorem}
\label{thm:algotwo}
Let $l<h$ be such that $h/l\in\Omega(n)$. There exists a polynomial-time, monotone, constant-fraction welfare-preserving transformation.
\end{theorem}

\begin{proof}
First we give a high-level intuition of the transformation. A monotone transformation needs to ensure that for any two adjacent inputs, it does not simultaneously occur that a 0 appears on $h$ and a 1 on $l$ in the differing position. As such, we would like to use the downward-closedness to ``zero out'' the $l$'s in a given input to avoid the undesirable situation. If the algorithm already returns a 1 on some $h$ for the input, this can be done while still preserving a constant fraction of the welfare. Otherwise, we look at nearby inputs and take an allocation that would return a 1 on some $h$ for our input, if such an allocation exists.

We now formally describe the transformation $\mathcal{T}$. Given an input $\textbf{v}$, $\mathcal{T}$ proceeds as follows:

\begin{enumerate}
\item If $\mathcal{A}(\textbf{v})$ already has a 1 on $h$, ``zero out'' all the $l$'s, and return that allocation. 
\item Else, if some input adjacent to $\textbf{v}$ has an allocation that would yield a 1 on $h$ at $\textbf{v}$, take that allocation and zero out all the $l$'s, and return that allocation. (Pick arbitrarily if there are many such allocations.)
\item Else, if some input of Hamming distance 2 away from $\textbf{v}$ has an allocation that would yield a 1 on $h$ at $\textbf{v}$, take that allocation and zero out all the $l$'s, and return that allocation. (Pick arbitrarily if there are many such allocations.)
\item Else, return $\mathcal{A}(\textbf{v})$.
\end{enumerate} 

The transformation takes polynomial time, and it only zeroes out the $l$'s when the allocation already has a 1 on $h$. Since $h/l=\Omega(n)$, a constant fraction of the welfare is preserved pointwise.

It remains to show that the resulting allocation rule is monotone. Suppose for contradiction that for some neighboring inputs $\textbf{v}$ and $\textbf{w}$, at the position where the two inputs differ, there exists a 0 on $h$ at $\textbf{v}$, and a 1 on $l$ at $\textbf{w}$. The allocation at $\textbf{w}$ cannot have changed in Steps 1, 2, or 3 of the transformation, and $\textbf{w}$ has 0 on all the $h$'s. But then $\textbf{v}$ must have started with 0 on all the $h$'s, except possibly at the position where the two inputs differ, because otherwise $\textbf{w}$ would have changed in Step 2. At the differing position, however, $\textbf{v}$ must have started with $0$ too, because otherwise it could never become 0. Now, $\textbf{v}$ must have changed in Step 2, because the allocation at $\textbf{w}$ satisfies the criterion in that step. It did not change to the allocation at $\textbf{w}$, because otherwise the non-monotonicity would not have occurred. Hence it must have changed to some other input with a 1 on $h$. But then $\textbf{w}$ should have changed to that allocation too in Step 3, a contradiction. \qed
\end{proof}

Note that the transformation in Theorem \ref{thm:algotwo} might preserve full welfare at a very small number of inputs. Indeed, if $\mathcal{A}$ returns the allocations with all 1's at every input, then $\mathcal{T}$ preserves full welfare at only 2 out of the $2^n$ inputs. Nevertheless, we can improve the transformation so that not only does it preserve a constant fraction of the welfare pointwise, but it also preserves full welfare at a $1/n$ fraction of the inputs. To this end, we will need to make a slightly stronger assumption that $h/l>n$.

\begin{theorem}
\label{thm:algotwobetter}
Let $l<h$ be such that $h/l>n$. There exists a polynomial-time, monotone, constant-fraction welfare-preserving transformation that preserves the full welfare at a $1/n$ fraction of the inputs.
\end{theorem}

\begin{proof}
We exhibit such a transformation $\mathcal{T}$, which is a slight modification of the transformation in Theorem \ref{thm:algotwo}. 

We call an allocation (implicitly along with an input) an \textit{$h$-allocation} if it has a 1 on $h$ at the input, and an \textit{$l$-allocation} otherwise. For any allocation (again implicitly along with an input), call another allocation a \textit{higher $h$-allocation} if it yields strictly more 1's on $h$ than the original allocation at the input.

Given any input $\textbf{v}$, the transformation $\mathcal{T}$ proceeds as follows:

\begin{enumerate}
\item If $\mathcal{A}(\textbf{v})$ is an $h$-allocation, consider its adjacent inputs. If the allocation at one of these inputs would yield a higher $h$-allocation at $\textbf{v}$, take that allocation. (Pick arbitrarily if there are many such allocations.)

\item Simulate Step 1 for all inputs of Hamming distance 1 and 2 away from $\textbf{v}$.

\item If the allocation at $\textbf{v}$ is an $l$-allocation, consider its adjacent inputs. If the allocation at one of these inputs would yield an $h$-allocation at $\textbf{v}$, take that allocation. (Pick arbitrarily if there are many such allocations.)

\item If the allocation at $\textbf{v}$ is still an $l$-allocation, consider the inputs of Hamming distance 2 away from $\textbf{v}$. If the allocation at one of these inputs would yield an $h$-allocation for $\textbf{v}$, take that allocation. (Pick arbitrarily if there are many such allocations.)

\item If the allocation at $\textbf{v}$ has improved to a higher $h$-allocation than the original allocation, zero out all the $l$'s.

\item Simulate Steps 1 through 5 for all inputs adjacent to $\textbf{v}$. Call the allocations at this point \textit{provisional allocations}. 

\item For any 1 on $l$, zero it out only if it yields a monotonicity conflict with the provisional allocation at a neighboring input.
\end{enumerate}

The transformation takes polynomial time. One can check in a similar way as in Theorem \ref{thm:algotwo} that the resulting allocation rule is monotonic, and that a constant fraction of the welfare is preserved pointwise. We now show that a $1/n$ fraction of the inputs obtain weakly better welfare. In particular, for each input that obtains strictly less welfare, we will find a neighbor that obtains weakly better (in fact, strictly better) welfare.

An input obtains strictly less welfare only if it has to zero out an $l$ in Step 7. That means that the input has a 1 on $l$. In particular, its allocation has never been changed in Steps 1 through 6. On the other hand, a neighbor has a provisional allocation with a 0 on $h$ in that position. Assume, for contradiction, that the neighbor obtains less (or equal) welfare than before. That means that it has never changed to a better allocation during the execution of the transformation. But then one of the two inputs could have gotten strictly more $h$'s by taking the allocation of the other, a contradiction.

Hence, every time an input loses a 1 on $l$, it can point to a neighbor that got better. Each input that got better can be pointed to at most $n-1$ times. Let $W$ be the set of inputs that got worse. We have $|W|\leq (n-1)\cdot(2^n-|W|)$, and therefore $|W|\leq\frac{n-1}{n}\cdot 2^n$, as desired. \qed
\end{proof}

If $h/l>2n$, the transformation in Theorem \ref{thm:algotwobetter} also preserves the expected welfare over the uniform distribution over the $2^n$ inputs, as we show next.

\begin{theorem}
\label{thm:algotwoexpected}
Let $l<h$ be such that $h/l>2n$. There exists a polynomial-time, monotone, constant-fraction welfare-preserving transformation that preserves full welfare at a $1/n$ fraction of the inputs and preserves expected welfare over the uniform distribution over the $2^n$ inputs.
\end{theorem}

\begin{proof}
Consider the transformation in Theorem \ref{thm:algotwobetter}. Every time an input loses a 1 on $l$, it can point to a neighbor that got better. The welfare of that neighbor has increased by at least $h-nl>nl$. Since each input that got better can be pointed to at most $n-1$ times, the expected welfare over the uniform distribution over the $2^n$ inputs is preserved. \qed
\end{proof}

Finally, we consider the other extreme case where $h/l$ is constant. In this case, simply returning a constant allocation already preserves a constant fraction of the approximation ratio. We focus on the allocation $\mathcal{A}(ll\dots l)$, but a similar statement can be obtained for any other constant allocation. The result can also be extended to the case where we have multiple input values, all of which are within a constant factor of each other.

\begin{theorem}
\label{thm:constalloc}
Let $l<h$ be arbitrary values (possibly depending on $n$), and let $\mathcal{T}$ be a transformation that returns the constant allocation $\mathcal{A}(ll\dots l)$ at any input. Then $\mathcal{T}$ preserves an $l/h$ fraction of the approximation ratio.
\end{theorem}

\begin{proof}
One can check that $\mathcal{T}(\mathcal{A})(\textbf{v})\geq \mathcal{A}(ll\dots l)$ for any input $\textbf{v}$. Moreover, we have that $OPT(\textbf{v})\leq \frac{h}{l}\cdot OPT(ll\dots l)$, since any allocation at $\textbf{v}$ would return at least an $l/h$ fraction of the welfare when allocated to the input $ll\dots l$. Hence \begin{align*}
approx_\mathcal{F}(\mathcal{T}(\mathcal{A}))&=\min_\textbf{v}\frac{\mathcal{T}(\mathcal{A})(\textbf{v})}{OPT_\mathcal{F}(\textbf{v})}\\
&\geq \frac{l\cdot \mathcal{A}(ll\dots l)}{h\cdot OPT_\mathcal{F}(ll\dots l)}\\
&\geq \frac{l}{h}\cdot approx_\mathcal{F}(\mathcal{A}),
\end{align*}
 as desired. \qed
\end{proof}

Combining this theorem with Theorem \ref{thm:algotwo}, we have that a constant fraction of the approximation ratio can be preserved if either $h/l$ is constant or $h/l\in\Omega(n)$. This means that if we were to obtain a negative result with two values, it would have to be the case that $h/l$ lies strictly between constant and linear.

\subsection{Multiple values}

In this subsection, we show that we can generalize the transformation in Theorem \ref{thm:algotwo} to the case where we have multiple input values, each pair separated by a ratio of $\Omega(n)$. Recall that when some two input values are separated by $O(n^\alpha)$ for some $\alpha\in [0,1)$, we have from Theorem \ref{thm:nonconstratioconstpointwise} that it is impossible to preserve a constant fraction of the welfare pointwise. Hence we have a complete picture of constant-fraction welfare-preserving transformations for multiple input values as well.

\begin{theorem}
\label{thm:algomorebetter}
Let $k$ be a constant, and let $a_1,\dots,a_k$ be such that $a_{i+1}/a_i\in\Omega(n)$ for $i=1,\ldots,k-1$. There exists a polynomial-time, monotone, constant-fraction welfare-preserving transformation. 

Moreover, if $a_{i+1}/a_i>n$ for all $i$, then the transformation can be modified so that it also preserves full welfare at a $1/(k-1)n$ fraction of the inputs.
\end{theorem}

\begin{proof}
We first consider the case where there are three input values $h,m,l$, and focus only on preserving a constant fraction of the welfare pointwise. It is possible to extend to any constant number of inputs $k$ and also preserve full welfare for a $1/(k-1)n$ fraction of the inputs, and we explain that later.

For any allocation (implicitly along with an input), we call it an \textit{$h$-allocation} if it has a 1 on $h$ at the input. Otherwise, we call it an \textit{$m$-allocation} if it has a 1 on $m$ at the input. Finally, we call it an \textit{$l$-allocation} if it is neither an $h$-allocation nor an $m$-allocation. For any allocation (again implicitly along with an input), call another allocation a \textit{higher allocation} if it yields either strictly more 1's on $h$ than the original allocation at the input, or an equal number of 1's on $h$ and strictly more 1's on $m$.

We exhibit a transformation $\mathcal{T}$ that preserves a constant fraction of the welfare pointwise. Given any input $\textbf{v}$, the transformation $\mathcal{T}$ proceeds as follows:

\begin{enumerate}
\item If $\mathcal{A}(\textbf{v})$ is an $l$-allocation, and some input adjacent to $\textbf{v}$ has an allocation that would yield an $m$-allocation or an $h$-allocation at $\textbf{v}$, or if $\mathcal{A}(\textbf{v})$ is currently an $m$-allocation, and some input adjacent to $\textbf{v}$ has an allocation that would yield an $h$-allocation at $\textbf{v}$, take that allocation for the time being. (Pick arbitrarily if there are many such allocations.)
\item If the allocation at $\textbf{v}$ is currently an $l$-allocation, and some input at Hamming distance 2 away from $\textbf{v}$ has an allocation that would yield an $m$-allocation at $\textbf{v}$, take that allocation for the time being.  (Pick arbitrarily if there are many such allocations.)
\item If the allocation at $\textbf{v}$ is currently \emph{not} an $h$-allocation, and some input at Hamming distance 3 away from $\textbf{v}$ has an allocation that would yield an $h$-allocation at $\textbf{v}$, take that allocation for the time being. (Pick arbitrarily if there are many such allocations.)
\item If the allocation at $\textbf{v}$ is currently an $m$-allocation, and some input at Hamming distance 4 away from $\textbf{v}$ has an allocation that would yield an $h$-allocation at $\textbf{v}$, take that allocation for the time being. (Pick arbitrarily if there are many such allocations.) 
\item If the allocation at $\textbf{v}$ is currently an $l$-allocation, and some input at Hamming distance 5 away from $\textbf{v}$ has an allocation that would yield an $h$-allocation at $\textbf{v}$, take that allocation for the time being. (Pick arbitrarily if there are many such allocations.)
\item If the allocation at $\textbf{v}$ is currently an $h$-allocation, zero out all the $m$'s and $l$'s. If it is an $m$-allocation, zero out all the $l$'s. Return the current allocation $\mathcal{A}(\textbf{v})$.
\end{enumerate}
The transformation runs in polynomial time. One can check in a similar way as in Theorem \ref{thm:algotwo} that the resulting allocation rule is monotone. Moreover, since $h/m,m/l\in\Omega(n)$, a constant fraction of the welfare is preserved pointwise.

As mentioned, it is possible to extend the transformation to any constant number of inputs $k$ and also preserve full welfare for a $1/(k-1)n$ fraction of the inputs. Suppose that the input values are $a_1<a_2<\dots<a_k$. Then the transformation takes $O(k^2)$ steps.
\begin{itemize}
\item $?\rightarrow ?$
\item $a_1\rightarrow a_2$
\item $?\rightarrow a_3$
\item $a_2\rightarrow a_3$
\item $a_1\rightarrow a_3$
\item $?\rightarrow a_4$
\item $a_1\rightarrow a_4$
\item $a_2\rightarrow a_4$
\item $a_3\rightarrow a_4$
\item $\dots$
\item $a_{n-1}\rightarrow a_n$
\end{itemize}

In each step, the transformation considers allocations at inputs at Hamming distance one higher than the previous step. If the change in the type of allocation (e.g., from an $a_2$-allocation to an $a_5$-allocation) matches the specified change in that step, the transformation executes the change. The question mark (e.g., $?\rightarrow a_3$) denotes \emph{any} allocation. Finally, the transformation zeroes out all the input values other than the highest one of the allocation. One can check that this transformation preserves a constant fraction of the welfare pointwise. We can extend it in a similar way as in Theorem \ref{thm:algotwobetter} so that the transformation also preserves full welfare at a $1/(k-1)n$ fraction of the inputs. \qed
\end{proof}

\subsubsection*{Acknowledgments.} The author thanks Tim Roughgarden for helpful discussion and acknowledges support from a Stanford Graduate Fellowship.

\newpage



\end{document}